\documentclass[twocolumn]{autart}

\usepackage[utf8]{inputenc}
\usepackage{amsmath,amssymb,amstext,mathtools,amsfonts,bm}
\usepackage{placeins}
\usepackage{graphicx}
\usepackage{epstopdf}
\usepackage{epsfig}
\usepackage[dvipsnames]{xcolor}
\usepackage{enumitem}
\usepackage{empheq}
\usepackage{tabto}
\usepackage{multirow}
\usepackage{tabularx}
\usepackage[misc]{ifsym}
\usepackage{epstopdf}
\usepackage{pifont}
\usepackage{arydshln}
\usepackage{comment}

\usepackage{upgreek}
\usepackage{tcolorbox}
\newtcolorbox{framedBox}{
notitle,
colframe=ForestGreen!10!white,
colback=ForestGreen!10!white,
boxsep=0pt,
top=0.1em,
bottom=0.1em,
left=0pt,
right=0pt,
middle=0pt,
sharp corners=northwest
}

\usepackage{natbib}

\newcommand{\mt}[1]{\bm{\mathsf{#1}}}
\newcommand{\vc}[1]{\bm{\mathsf{#1}}}

\newcommand{\sign}[1]{\ensuremath{\operatorname{sign}\left( {#1} \right)}}

\newcommand{\sat}[1]{\text{sat}\left( #1 \right)}

\newcommand{\sgnpow}[2]{\left\lfloor #1 \right\rceil^{#2}}

\newcommand{\sgn}[1]{\sgnpow{#1}{0}}

\newcommand{\nonum}{\nonumber \\}

\definecolor{dark-green}{HTML}{009B55}

\newcommand{\boldphi}{\boldsymbol{\varphi}}
\newcommand{\boldgamma}{\boldsymbol{\gamma}}
\newcommand{\boldf}{\boldsymbol{f}}

\newcommand{\boldS}{\boldsymbol{S}}
\newcommand{\boldd}{\boldsymbol{d}}

\newcommand{\boldT}{\boldsymbol{T}}

\newcommand{\boldzeta}{\boldsymbol{\zeta}}

\newcommand{\boldalpha}{\boldsymbol{\alpha}}
\newcommand{\boldbeta}{\boldsymbol{\beta}}

\newcommand{\boldA}{\boldsymbol{A}}

\newcommand{\observO}{\boldsymbol{\mathcal{O}}}
\newcommand{\bolde}{\boldsymbol{e}}
\newcommand{\boldt}{\boldsymbol{t}}
\newcommand{\boldtheta}{\boldsymbol{\theta}}
\newcommand{\boldzero}{\boldsymbol{0}}

\makeatletter
\DeclareRobustCommand{\qed}{%
  \ifmmode
    \eqno \def\@badmath{$$}%
    \let\eqno\relax \let\leqno\relax \let\veqno\relax
    \hbox{$\square$}%
  \else
    \leavevmode\unskip\penalty9999 \hbox{}\nobreak\hfill
    \quad\hbox{$\square$}%
  \fi
}
\makeatother

\usepackage{manfnt}

\def\boldalpha{{\boldsymbol{\alpha}}}
\def\boldgamma{{\boldsymbol{\gamma}}}
\def\boldGamma{{\boldsymbol{\Gamma}}}
\def\boldphi{{\boldsymbol{\varphi}}}

\def\boldzero{{\boldsymbol{0}}}

\def\boldzeta{{\boldsymbol{\zeta}}}
\def\boldtheta{{\boldsymbol{\theta}}}
\def\boldbeta{{\boldsymbol{\beta}}}

\def\boldPhi{{\boldsymbol{\Phi}}}

\newcommand{\propertransformedstate}{\eta}

\renewcommand{\epsilon}{\varepsilon}
\renewcommand{\phi}{\varphi}

\def\A{\mt{A}}
\def\B{\mt{B}}

\def\I{\mt{I}}
\def\J{\mt{J}}

\def\M{\mt{M}}
\def\N{\mt{N}}

\def\S{\mt{S}}

\def\a{\vc{a}}

\def\d{\mathrm{d}}
\def\e{\vc{e}}
\def\f{\vc{f}}

\def\k{\vc{k}}

\def\v{\vc{v}}

\begin{document}

\begin{frontmatter}

\title{Model-Based Disturbance Rejection via \\Sliding-Mode Control: Continuous-Time Formulation \\and
Proper Implicit Discretization\thanksref{moco}}

\author[First]{Benedikt Andritsch}\ead{benedikt.andritsch@tugraz.at},
\author[First]{Stefan Koch}\ead{stefan.koch@tugraz.at},
\author[Second]{Jaime A. Moreno}\ead{JMorenoP@ii.unam.mx}

\address[First]{Institute of Automation and Control at Graz University
of Technology, Graz, Austria.}
\address[Second]{Instituto de Ingenier\'{i}a, Universidad Nacional Aut\'onoma de M\'exico (UNAM), Coyoac\'an, D.F., 04510, Mexico}

\thanks[moco]{This work was supported by the Horizon Europe research and innovation programme under the Marie Skłodowska-Curie Grant agreement
101183051.}


 \begin{abstract}
We propose a robust exact sliding-mode controller for reference tracking with model-based disturbance rejection, where the disturbance is generated by a linear time-invariant model driven by a bounded input and may be bidirectionally coupled %
with the plant. 
The controller embeds the internal model principle and its error dynamics mirror those of the robust exact differentiator, ensuring finite-time convergence. 
We develop a proper implicit discretization for sampled-data implementation that preserves stability and accuracy. In the nominal case, which contains a disturbance generated by a $p$-th order %
integrator chain, %
the method achieves a control accuracy of $T^{p+1}\bar D$, with discretization time $T$ and disturbance bound $\bar D$.
This accuracy is shown to relate to the fundamental lower bound %
by a factor depending only on the order of the disturbance model. 
Simulations demonstrate the efficacy of the proposed methods.%
\end{abstract}

\end{frontmatter}

\section{Introduction}
Sliding-mode control (SMC) is an established robust control method against matched disturbances that exploits discontinuities in the controller to stabilize the system state on a predefined sliding surface in finite time, i.e., drive a sliding variable to zero, see~\cite{SLOTINE1984,Utkin1992}.
To mitigate chattering phenomena in case of sliding variables with relative degree one, higher-order sliding-mode (HOSM) control (HOSMC) was introduced that lifts the discontinuity to derivatives of the control input, see~\cite{Levant1993,Bartolini1998,Levant2003}. 
A side-effect of these HOSMC concepts is that the class of disturbances that can be rejected by these controllers changes compared to the classical first-order SMC.
In particular, first-order SMC is robust against bounded disturbances, whereas second-order SMC like the super-twisting algorithm (STA) proposed by~\cite{Levant1993} is robust against Lipschitz continuous disturbances.
This Lipschitz-continuous disturbance of the STA can be understood to be generated by an integrator driven by a bounded disturbance input, i.e., a first-order disturbance model.
Over the years, several extensions and generalizations of the STA have been proposed, that all assume some Lipschitz-conditions of the disturbance, see e.g.,~\cite{Shtessel2012,Kamal2014,%
Castillo_2017,RuedaEscobedo2024,Fornaro_2025}.
Also for other HOSM algorithms, variants have been presented that assume Lipschitz-continuity properties of the disturbance, like the continuous twisting algorithm by~\cite{TorresGonzalez2015} and continuous nested sliding mode algorithms by~\cite{Fridman_2015}.
So, HOSMC schemes that generate continuous control signals implicitly incorporate a disturbance model, represented by the class of disturbances that can be rejected.
In parallel, a sliding-mode differentiator, called the robust exact differentiator (RED), was developed by~\cite{Levant1998,Levant2003} to estimate derivatives of a signal in finite time.
The RED incorporates a higher-order disturbance model consisting of a chain of integrators driven by a bounded disturbance input.
Both, the sliding-mode differentiator and continuous HOSMC schemes, therefore apply the internal model control principle in the form of the disturbance model. %

A topic that received a lot of attention in the past is the discrete-time implementation of sliding mode controllers.
It is well-known that applying the straight-forward explicit Euler discretization to SMC leads to high-frequency oscillations, commonly known as discretization chattering, in the steady state, deteriorating the control performance.
Many publications were dedicated to analyzing discretized sliding-mode algorithms and developing discretization methods that avoid this form of chattering, see, e.g.,~\cite{Levant2011,%
Livne_2014,Brogliato2018,Koch2019,%
Hanan2022,Andritsch2024,Seeber_2025}.

In this paper we generalize the application of the internal model control principle in the design of sliding-mode controllers for sliding variables with relative degree one. 
We assume a disturbance generated by an arbitary linear time-invariant system driven by a bounded disturbance input and by the sliding variable itself. %
This disturbance may be unbounded and non-Lipschitz continuous, in comparison to existing approaches. %
A sliding-mode based control concept  is proposed %
that robustly stabilizes the sliding variable in finite time. %
The resulting closed-loop tracking error dynamics %
are equivalent to those of the arbitrary-order RED, %
leading to the denotion as robust exact sliding-mode controller (RESMC).
The continuous-time controller design in this paper is a variation to the design presented by~\cite{AndritschCDC2025}. 
Furthermore, a proper implicit discretization of the RESMC is presented that avoids discretization chattering effects due to sampled-data control.
After introducing the control problem (Section~\ref{sec:problem}), the contributions of this paper are:
\begin{itemize}
    \item design of the RESMC for reference tracking (Section~\ref{sec:remsc});
    \item derivation of a system discretization (Section~\ref{plant-discretization}), that is sample-wise exact in the nominal case;
    \item development of a proper implicit discretization of the RESMC that avoids discretization chattering (Section~\ref{sec:implicit-discretization});
    \item establishing a fundamental lower bound for the best possible worst-case control accuracy (Section~\ref{sec:control-accuracy}).
\end{itemize}
Simulations are demonstrated in Section~\ref{sec:simulations}.
Finally, the paper is concluded in Section~\ref{sec:conclusion}.

\textbf{Notation:}
The sign function is defined as \\${\sign{x} = \begin{cases}
\frac{x}{|x|} &\text{if } x\neq 0,\\
0 & \text{else,}
\end{cases}}$ with $x\in\mathbb{R}$. 
The sign-preserving power function $\sgnpow{x}{a} = \sign{x}|x|^a, x,a\in\mathbb{R}, a\neq 0$ is used, whereas $\sgn{x} = \begin{cases}
\{\sign{x}\} &\text{if } x\neq0,\\
[-1,1] &\text{else}
\end{cases}$ denotes the set-valued sign function.
Bold capital letters refer to matrices, bold lowercase letters to vectors.
$\I$ denotes the identity matrix, $\e_i$ denotes the unit vector with a single non-zero value of one in the $i$-th component, of appropriate sizes.
Given a vector $\a\in\mathbb{R}^{n\times 1}$, %
$a_i$ refers to the $i$-th entry of $\a$.

\section{Problem Description}\label{sec:problem}

In this paper we consider the reference tracking problem of a scalar signal $y$ following the linear disturbed dynamics
\begin{align}\label{output-dynamics}
\dot y = a y + b (u(t)+w(t)),
\end{align}
with known and constant parameters $a,b \in \mathbb{R}$, with $b\neq 0$, control input $u(t) \in \mathbb{R}$ and disturbance $w(t) \in \mathbb{R}$.
The control goal is that the output $y(t)$ tracks the reference signal $r(t) \in \mathbb{R}$ that is differentiable at least once.
The reference $r(t)$ as well as its derivative $\dot r(t)$ are given. 
The dynamics of the tracking error $\sigma(t) = y(t) - r(t)$ are
\begin{align}\label{output-equation}
\dot \sigma &=  a y - \dot r + b (u+w).
\end{align}

\begin{rem}
   Often, the signal $y(t)$ is the output with relative degree one of a higher-order plant. 
   Then, the states of the internal dynamics of this plant must stay bounded when $y(t)$ tracks $r(t)$. 
   For details, see, e.g.,~\cite{Byrnes_1989,Khalil_2002,Shtessel2014}.
   In this paper, however, we focus on the simple first-order setup in~\eqref{output-dynamics}.
   The results can be transferred to controlling the relative degree one output of a higher-order system.
\end{rem}

We assume that the disturbance $w(t)$ can be modeled by the known linear time-invariant system
\begin{align}\label{disturbance-model}
   \dot\boldphi &= \boldS \boldphi + \boldf y + 
   \boldd \Delta(t), \nonum
   w &= \boldgamma^\top \boldphi,
\end{align}
with disturbance state $\boldphi(t)\in\mathbb{R}^p$, matrix $\boldS \in \mathbb{R}^{p\times p}$, vectors $\boldf,\boldd,\boldgamma\in\mathbb{R}^p$, and bounded Lebesgue-measurable unknown input $\Delta(t)$, i.e. $|\Delta(t)| \leq D\,\,\forall t\geq 0$.
We assume that the output $w(t)$ of the  disturbance model~\eqref{disturbance-model} has %
relative degree $p$ with respect to $\Delta(t)$.%

\begin{rem}\label{rem:stc-and-integrator-chain}
   The well-known super-twisting controller by~\cite{Levant1993} is a controller for a special case of the presented setting with an integrator as disturbance model, i.e.~\eqref{disturbance-model} with $p=1$, $\boldS = 0$, $\boldf = 0$, $\boldd = 1$, and $\boldgamma = 1$.
   Then, $w(t)$ is a Lipschitz-continuous function with a bounded derivative.
   Also, the more general form of a signal $w(t)$ with a bounded $p$-th derivative, i.e. $|w^{(p)}| \leq D$, can be expressed via the disturbance model~\eqref{disturbance-model} with arbitrary $p$, $\boldS = \boldzero$, $\boldf = \boldzero$, $\boldd = \e_p$, and $\boldgamma = \e_1$.
\end{rem}

In the first part of this paper, we develop a continuous-time controller for the proposed control problem.
Afterwards, in the main part of this work, we assume that the control is applied to the plant via a  zero-order hold element with constant hold time, i.e.
\begin{align}\label{zoh-control}
   u(t) &= u_k \quad \forall t\in [kT, (k+1)T),\nonum
   y_k &= y(kT),
\end{align}
with discretization time $T>0$, discrete control sequence $u_k, k\in\mathbb{Z}^+_0$, and samples $y_k$ of the system output $y(t)$ at the same time-instances $kT$ as the discrete-time control $u_k$ is applied.
Only these samples $y_k$, and not the entire continuous-time function $y(t)$, are available to the controller.
This control setting is known as sampled-data control in literature, see, e.g.,~\cite{Franklin_1997}.

\subsection{Transformation to Observable Canonical Form}\label{sec-observable-canonical-form}
In the end, the closed-loop error dynamics will be  equivalent to the error dynamics of the RED.
This relation to the RED, which is a sliding-mode observer, is the motivation for transforming the system to observable canonical form.
By combining~\eqref{output-equation} and~\eqref{disturbance-model}, the overall system dynamics can be written as
\begin{align}\label{combined-plant}
   \begin{bmatrix}
\dot\sigma \\ \dot\boldphi
\end{bmatrix} &= \boldA \begin{bmatrix}
\sigma \\ \boldphi
\end{bmatrix} + \begin{bmatrix}
   0 \\ \boldf
\end{bmatrix} r  + \begin{bmatrix}
    1\\\boldzero
\end{bmatrix} (ay - \dot r + b u) + \begin{bmatrix}
0\\\boldd
\end{bmatrix} \Delta,
\end{align}
with
\begin{align}
\boldA  = \begin{bmatrix}
0 & \bar\boldgamma^\top\\
\boldf & \boldS
\end{bmatrix}, \quad \bar\boldgamma = b\boldgamma.
\end{align}
From the relative degree $p$ of $w(t) = \boldgamma^\top\boldphi(t)$ with respect to $\Delta(t)$ and~\eqref{output-equation} we know that $\frac{\sigma(t)}{b}$ has relative degree $p+1$ with respect to $\Delta(t)$. Therefore, system~\eqref{combined-plant} is strongly observable, see~\cite{Hautus_1983}. %
In the following, we transform the linear system~\eqref{combined-plant} to observable canonical form.
The transformed system will later on be used to design the controller.
Recalling the notation $\e_1 = \begin{bmatrix}1&0&\dots &0\end{bmatrix}^\top$, the observability matrix of~\eqref{combined-plant} with output $\sigma$ is given by
\begin{align}\label{observabilitycanonicalform-start}
\observO = \begin{bmatrix}
   \bolde_1^\top\boldA^0\\
   \vdots\\
   \bolde_1^\top\boldA^p
\end{bmatrix}.
\end{align}
From strong observability of~\eqref{combined-plant} follows
\begin{align}
\observO\begin{bmatrix}
0\\\boldd
\end{bmatrix} = \begin{bmatrix}
   0 & \dots & 0 & \mu
\end{bmatrix}^\top,
\end{align}
as described by~\cite{Niederwieser2019}, with
\begin{align}
\mu = \bolde_1^\top\boldA^p \begin{bmatrix}
0\\\boldd
\end{bmatrix} \neq 0.
\end{align}
Define the transformation matrix
\begin{align}
   \boldT = \begin{bmatrix}
      \boldA^p\boldt & \dots & \boldA\boldt & \boldt
   \end{bmatrix},
\end{align}
with the last column of the inverse observability matrix
\begin{align}\label{observabilitycanonicalform-end}
\boldt = \observO^{-1} \bolde_{p+1}.
\end{align}
We transform the states according to
\begin{align}\label{transformation-zeta}
   \boldzeta = \boldT^{-1} \begin{bmatrix}
   \sigma\\\boldphi
   \end{bmatrix}.
\end{align}
We define the vectors
\begin{align}\label{beta1-0}
   \boldbeta = \boldT^{-1}\begin{bmatrix}
0\\\boldf
\end{bmatrix},\,\, \boldtheta = \boldT^{-1}\begin{bmatrix}
1\\\boldzero
\end{bmatrix},\text{ with } \theta_1 = 1,\,\, \beta_1 = 0.
\end{align} 
Using~\eqref{combined-plant},~\eqref{transformation-zeta} and~\eqref{beta1-0} this yields the dynamics in observable canonical form
\begin{align}\label{zeta-system}
   \dot\boldzeta &= \begin{bmatrix}
      -\alpha_1 & 1& 0 & \dots & 0\\
      -\alpha_2 & 0 & 1 & \dots & 0\\
      \vdots & & \ddots & \ddots & \vdots\\
      -\alpha_{p+1} & 0 & 0 & \dots & 0
   \end{bmatrix} \boldzeta + \boldbeta r(t) + \nonum
   &\qquad+\boldtheta \left(bu(t)-\dot r(t) + ay(t)\right) + \begin{bmatrix}\boldzero\\\delta(t)\end{bmatrix},\nonum
   \sigma(t) &= \zeta_1(t),
\end{align}
with some coefficients $\alpha_i\in\mathbb{R}$, and disturbance $\delta(t) = \mu\Delta(t)$ fulfilling
\begin{align}
|\delta(t)| \leq \mu D = \bar D,\, \forall t\geq 0.
\end{align}
The %
problem of making $y(t)$ track $r(t)$ is equivalent to driving $\zeta_1(t)$ of the transformed system~\eqref{zeta-system} to zero.
In the following sections we design the control input $u(t)$ to achieve this goal.

\section{Continuous-Time Robust Exact Sliding-Mode Controller}\label{sec:remsc}

In this section we derive the continuous-time robust exact sliding-mode controller (RESMC). 
This will lead to a continuous-time control signal $u(t)$.
The controller has access to the reference $r(t)$ and the output $y(t)$, and therefore also to the output error $\sigma(t) = \zeta_1(t) = y(t) - r(t)$.
The disturbance $w(t)$ and the state $\boldphi(t)$ of the disturbance model~\eqref{disturbance-model} are not available.

We define the control law
\begin{align}\label{continuous-time-control-law}
    u(t) = \frac{1}{b}\left(\bar u(t) + \dot r(t) - a y(t)\right),
\end{align}
that compensates the known plant dynamics 
and construct the dynamic controller
\begin{align}\label{dynamic-controller}
\bar u(t) &= c_1(t) + \psi_1(\zeta_1) + v_1,\nonum
\dot v_i &= v_{i+1} + c_{i+1}(t) + \psi_{i+1}(\zeta_1),\,\forall i\in\{1,\dots,p-1\},\nonum
\dot v_p &= c_{p+1}(t) + \psi_{p+1}(\zeta_1),
\end{align}
with controller states $v_i(t) \in \mathbb{R}$, and controller functions $c_i(t) \in \mathbb{R}$ and $\psi_i(\zeta_1) \in \mathbb{R}$ to be designed.
The functions $c_i(t)$ are responsible for compensating parts of the disturbance dynamics, whereas
the functions $\psi_i(\zeta_1)$ are responsible for stabilizing the error dynamics.
This controller~\eqref{dynamic-controller} copies the structure of the transformed disturbance dynamics consisting of the 2nd to ${(p+1)\text{-th}}$ row in~\eqref{zeta-system}.

Let the error states be defined as $\varepsilon_1 = \zeta_1$ and ${\varepsilon_i = \zeta_i + v_{i-1}}$, ${i\in\{2,\dots,p+1\}}$.
Combining~\eqref{zeta-system},~\eqref{continuous-time-control-law} and~\eqref{dynamic-controller} while considering $\theta_1 = 1$ and $\beta_1 = 0$ from~\eqref{beta1-0}  yields the error dynamics
\begin{align}\label{intermediate-error-dynamics}
   \dot\varepsilon_1 &= -\alpha_1\varepsilon_1 + \varepsilon_2 + c_1(t) + \psi_1(\varepsilon_1),\nonum
   \dot\varepsilon_i &= -\alpha_{i}\varepsilon_1 + \varepsilon_{i+1} +  \beta_{i}r(t) + \theta_{i}\bar u(t) + \nonum
   &\qquad + c_i(t) + \psi_{i}(\varepsilon_1),\,\forall i\in\{2,\dots,p\},\nonum
   \dot\varepsilon_{p+1} &= -\alpha_{p+1}\varepsilon_1 + \beta_{p+1}r(t) + \theta_{p+1}\bar u(t) + \nonum
   &\qquad + c_{p+1}(t) + \psi_{p+1}(\varepsilon_1) + \delta(t),\nonum
   \sigma(t) &= \varepsilon_1(t).
\end{align}
Now, we select the compensating controller functions
\begin{align}\label{controller-functions-c}
   c_1(t) &= \alpha_1\varepsilon_1,\nonum
   c_i(t) &= \alpha_i\varepsilon_1 - \beta_i r(t) - \theta_{i}\bar u(t), i \in \{2,\dots,p+1\},%
\end{align}
and stabilizing controller functions
\begin{align}\label{controller-functions-g}
   \psi_i(\varepsilon_1) &= - k_i\sgnpow{\varepsilon_1}{\frac{p+1-i}{p+1}}, i \in \{1,\dots,p\},\nonum
   \psi_{p+1}(\varepsilon_1) &\in - k_{p+1}\sgn{\varepsilon_1}.
\end{align}
Inserting the controller functions~\eqref{controller-functions-c} and~\eqref{controller-functions-g} into the error dynamics~\eqref{intermediate-error-dynamics} finally yields the closed-loop error dynamics
\begin{align}\label{continuous-error-dynamics}
   \dot\varepsilon_i &= -k_i \sgnpow{\varepsilon_1}{\frac{p+1-i}{p+1}} + \varepsilon_{i+1}, \,i\in\{1,\dots,p\},\nonum
   \dot\varepsilon_{p+1} &\in -k_{p+1} \sgn{\varepsilon_1} + \delta(t).
\end{align}

\begin{rem}
Note that the controller~\eqref{continuous-time-control-law},~\eqref{dynamic-controller} with controller functions~\eqref{controller-functions-c},~\eqref{controller-functions-g} is a variant of the controller developed by~\cite{AndritschCDC2025}. There, the original disturbance model~\eqref{disturbance-model} is incorporated in the controller. Here, the transformed disturbance model, being represented by an integrator chain and additional influence by the control error $\varepsilon_1$, is incorporated in the controller.
The two variants are equivalent via the state transformation in Section~\ref{sec-observable-canonical-form}.
\end{rem}

\begin{thm}\label{theorem:cont-resmc}
   Consider system~\eqref{output-dynamics} with disturbance model~\eqref{disturbance-model} controlled by~\eqref{continuous-time-control-law},~\eqref{dynamic-controller},~\eqref{controller-functions-c},~\eqref{controller-functions-g}. 
   There exist gains $k_i$, $i\in\{1,\dots,p+1\}$ such that the output $y(t)$ converges to the reference $r(t)$ in finite time and stays there, despite the bounded disturbance $\Delta(t)$.
   Furthermore, the controller states $v_i(t),\,i\in\{1,\dots,p\}$ converge to %
   $\zeta_{i+1}(t)$ from~\eqref{transformation-zeta}.
\end{thm}
\begin{pf}
   The closed-loop error dynamics~\eqref{continuous-error-dynamics} are identical to the error dynamics of the robust exact differentiator by~\cite{Levant2003,Levant2018} with bounded disturbance $\delta(t)$.
   Therefore, the gains $k_i$, $i\in\{1,\dots,p+1\}$ can be chosen such that $\varepsilon_i(t) = 0\,\forall i\in\{1,\dots,p+1\},\,\forall t>T_c$ for some finite time $T_c>0$ despite the bounded disturbance $\delta(t)$.
   Thus, $y(t) = \sigma(t) + r(t) = \varepsilon_1(t) + r(t) = r(t)$ holds $\forall t>T_c$, and $v_i(t) = \zeta_{i+1}(t),\,i\in\{1,\dots,p\}$.
   \qed
\end{pf}
\begin{cor}\label{cor-bounded-states}
   Consider system~\eqref{output-dynamics} with disturbance model~\eqref{disturbance-model} and controller~\eqref{continuous-time-control-law},~\eqref{dynamic-controller},~\eqref{controller-functions-c},~\eqref{controller-functions-g} with gains $k_i$ fulfilling Theorem~\ref{theorem:cont-resmc}. 
   Let the disturbance states $\boldphi(t)$ be bounded, i.e. $\exists M <\infty: {||\boldphi(t)|| < M}$, ${\forall t\geq0}$.
   Then, the states $\v(t)$ of the controller stay bounded as well, i.e., $\exists N <\infty: {||\v(t)|| < N}$, ${\forall t\geq0}$.
\end{cor}
\begin{pf}
   The controller dynamics~\eqref{dynamic-controller} have infinite escape time for bounded input functions $c_i(t)$, $\psi_i(\varepsilon_1)$.
   These functions are bounded due to~\eqref{controller-functions-c},~\eqref{controller-functions-g} and $\varepsilon_1$ being finite-time stable due to Theorem~\ref{theorem:cont-resmc}.
   From Theorem~\ref{theorem:cont-resmc} follows also that after finite time $v_i(t) = -\zeta_{i+1}(t)$ holds $\forall i\in\{1,\dots,p\}$ and $\sigma(t) = 0$. 
   The transformed states $\zeta_i$ result from a linear combination of $\sigma$ and $\boldphi$ through the transformation~\eqref{transformation-zeta}.
   As $\boldphi$ stays bounded, and $\sigma = 0$ after finite time, $\zeta_i$ and thus $v_i$ stay bounded as well.
   \qed
\end{pf}
Note that the boundedness assumption of the disturbance states $\boldphi(t)$ in Corollary~\ref{cor-bounded-states} does not hold in general, e.g., if the system matrix $\S$ in~\eqref{disturbance-model} is not Hurwitz.

In the following sections we restore the assumption of sampled-data control~\eqref{zoh-control}. 

\section{Plant Discretization}\label{plant-discretization}
In this section we derive discrete-time representations of the continuous-time transformed system dynamics~\eqref{zeta-system}.
At first, we consider the nominal case of a disturbed integrator chain, followed by the general case.
These discrete-time representations will be used later on to design discrete-time controllers.

\subsection{Exact Plant Discretization in the Nominal Case}\label{nominal-plant-discretization}
At first, we focus on the nominal case of reference tracking of an integrator that is disturbed by an integrator chain with unknown input via sampled-data control.
The corresponding system description~\eqref{output-equation},~\eqref{disturbance-model} therefore contains the parameters $a = 0$, $\f = \boldzero$, $\S = \begin{bmatrix}
    \boldzero & \I\\0 &\boldzero^\top
\end{bmatrix}$, $\boldgamma = \e_1$, and $\boldd = \e_p$.
Thus, the transformed system~\eqref{zeta-system} takes the form of an integrator chain with $a=0$, and $\alpha_i = \beta_i = 0\,\forall i\in\{1,\dots,p+1\}$.
With the sampled-data control~\eqref{zoh-control}, system~\eqref{zeta-system} then takes the form
\begin{align}\label{nominal-zeta-system}
   \dot\zeta_1(t) &= \zeta_2(t) + b u_k - \dot r(t) \quad\forall t\in [kT,(k+1)T),\nonum
   \zeta_2^{(p)}(t) &= \delta(t),\nonum
   \sigma(t) &= \zeta_1(t).
\end{align}
Note that $\zeta_1(t) = y(t) - r(t)$ and $\zeta_2(t) = bw(t)$ compared to~\eqref{output-dynamics} holds in this case. This means $\zeta_1(t)$ is the output error that should be stabilized despite the disturbance $\zeta_2(t)$ whose $p$-th derivative is bounded.
This control problem was discussed in Remark~\ref{rem:stc-and-integrator-chain}.
In this section, we find a discrete-time representation for system~\eqref{nominal-zeta-system}.

Ultimately, we want to be able to transfer stability arguments from discrete-time to the samples of the continuous-time trajectory.
Therefore, we demand the following property:
\begin{enumerate}
   \item[(P1)] the discrete-time system is sample-wise exact in $\zeta_1$ in the sense that $\zeta_1(kT) = \zeta_{1,k}\,\forall k\geq 0$.
\end{enumerate}
Additionally, the goal is to derive an implicit discretization of the controller.
This means, the discrete-time closed-loop in the end must correspond to an implicit Euler discretization of the continuous-time closed loop.
Thus, we furthermore demand the second property:
\begin{enumerate}
   \item[(P2)] the discrete-time system is in implicit Euler form.
\end{enumerate}
Finally, we want to ensure consistency regarding tuning of the discrete-time controller compared to the continuous-time one. Therefore, we require the third property:
\begin{enumerate}
   \item[(P3)] the discrete-time disturbance $\delta_{k+1}$ has the same bound as the continuous-time disturbance $\delta(t)$, i.e. $|\delta(t)| \leq \bar D\,\forall t>0\,\Rightarrow\,|\delta_{k+1}| \leq \bar D\,\forall k> 0$.
\end{enumerate}

\begin{lem}\label{lemma-nominal-discretization}
   There exists a discrete-time representation of~\eqref{nominal-zeta-system} that fulfills properties (P1)-(P3).
\end{lem}
\begin{pf}
The solution of the first equation of~\eqref{nominal-zeta-system} at the sampling instances $t=kT$ using the abbreviation $k^+ = k+1$ is
\begin{align}\label{solution-zeta1kp1}
&\zeta_1(k^+T) = \zeta_{1}(kT) + \int_{kT}^{k^+T} b u_k - \dot r(t) + \zeta_2(t) \,\d t = \nonum
&\quad =  \zeta_{1}(kT) + T b u_k - T\dot r_k %
+ \int_{kT}^{k^+T} \zeta_2(t) \,\d t,
\end{align}
with
\begin{align}
\label{dotrk}
&\dot r_k = \frac{1}{T}\int_{kT}^{k^+T}\dot r(t)\d t = \frac{r(k^+T) - r(kT)}{T}.
\end{align}
To fulfill property (P1), assume $\zeta_{1,i} = \zeta_1(iT)$ holds for $i \leq k$ and subsequently define
\begin{align}\label{definition-zeta1k}
\zeta_{1,k+1} \coloneqq \zeta_1((k+1)T),
\end{align}
with the solution~\eqref{solution-zeta1kp1}.
Further, to take into account property (P2), we define the discrete-time disturbance states as
\begin{align}\label{definition-zeta2k}
&\zeta_{2,k+1} = \frac{1}{T} \int_{kT}^{k^+T} \zeta_2(t) \,\d t,\nonum
&\zeta_{3,k+1} = \frac{\zeta_{2,k+1}-\zeta_{2,k}}{T} = \frac{1}{T^2} \int_{kT}^{k^+T} \zeta_2(t) - \zeta_2(t-T) \,\d t,\nonum
&\qquad\vdots\nonum
&\zeta_{p+1,k+1} = \frac{\zeta_{p,k+1} - \zeta_{p,k}}{T} = \nonum
&\hspace{-0.4em}= \frac{1}{T^p} \int_{kT}^{k^+T} \sum_{i=0}^{p-1}(-1)^i {{p-1} \choose {i}} \zeta_2(t - iT) \,\d t,\,\text{if } k\geq p-1,
\end{align}
with $\zeta_{i,0} = 0$ for $i\in\{2,\dots,p+1\}$, and the disturbance input as
\begin{align}\label{definition-deltak}
&\delta_{k+1} = \frac{1}{T}(\zeta_{p+1,k+1} - \zeta_{p+1,k}) = \nonum
&\quad= \frac{1}{T^{p+1}} \int_{kT}^{k^+T} \sum_{i=0}^{p}(-1)^i {{p} \choose {i}} \zeta_2(t - iT) \,\d t,\,\text{if } k\geq p.
\end{align}
Note that the integrand $\rho(t) = \sum_{i=0}^{p}(-1)^i {{p} \choose {i}} \zeta_2(t - iT)$ for all $t\in[kT,(k+1)T)$ on the right-hand side of $\delta_{k+1}$ for $k\geq p$ corresponds to the $p$-th order forward difference of the signal $\zeta_2(t)$. 
By recursive application of the identity $F(t) - F(t-T) = \int_{0}^{T} f(t-T+\tau) \,\d \tau$ for any function $f(t)$ and its antiderivative $F(t) = \int_0^t f(\tau) \d\tau$, we can rewrite and bound the forward difference by
\begin{align}
&\rho(t) = \int_0^T\dots \int_0^T \delta(t-pT +\tau_1+   \dots + \tau_p) \,\d \tau_1\dots \d \tau_p,\nonum
&\Rightarrow |\rho(t)| \leq \int_0^T\dots \int_0^T \bar D \,\d \tau_1\dots \d \tau_p = T^p \bar D.
\end{align}
Inserting this bound into the definition of $\delta_{k+1}$ in~\eqref{definition-deltak} yields $|\delta_{k+1}| \leq \bar D$, i.e., property (P3) is fulfilled.

Using definitions~\eqref{definition-zeta1k},~\eqref{definition-zeta2k}, and~\eqref{definition-deltak} directly yields the discrete-time system
\begin{align}\label{discrete-time-nominal-plant}
\zeta_{1,k+1} &= \zeta_{1,k} + T \zeta_{2,k+1} + T b u_k,\nonum
\zeta_{i,k+1} &= \zeta_{i,k} + T \zeta_{i+1,k+1},\,i\in\{2,\dots,p\},\nonum
\zeta_{p+1,k+1} &= \zeta_{p+1,k} + T \delta_{k+1},\nonum
\sigma(kT) &= \zeta_{1,k},
\end{align}
which can be seen as an implicit Euler discretization of the system~\eqref{nominal-zeta-system}, i.e., it fulfills (P2), and at the same time fulfills property (P1).
\qed
\end{pf}

\subsection{Approximative Plant Discretization in the General Case}

Consider now the general form of the transformed continuous-time system~\eqref{zeta-system}.
This corresponds to the reference tracking of the output that is disturbed by an arbitrary linear time-invariant system.
In the following, we derive a discrete-time representation of this system under the sampled-data control law~\eqref{zoh-control}.
In this case, however, we cannot find a sample-wise exact discrete-time representation as introduced as the first desired property in Section~\ref{nominal-plant-discretization}, because the system~\eqref{zeta-system} is not in pure feedback form, and we have no continuous knowledge of $\zeta_{1}(t)$. 
Instead, we derive a discrete-time approximation of~\eqref{zeta-system} that after compensation is in implicit Euler structure and has a bounded disturbance, i.e. fulfills the second and third properties from Section~\ref{nominal-plant-discretization}.

Using $y(t) = r(t) + \zeta_1(t)$ and the sampled-data control~\eqref{zoh-control}, system~\eqref{zeta-system} can be written as
\begin{align}\label{general-case-system}
    \dot\boldzeta &= \boldGamma(\boldzeta,t) + \J \boldzeta,\quad \forall t\in [kT, (k+1)T), \text{ with}\nonum
    \boldGamma(\boldzeta,t) &= \M\boldzeta + (\boldbeta + a\boldtheta) r(t) - \boldtheta \dot r(t) + \boldtheta b u_k + \e_{p+1} \delta(t),\nonum
    \M &= (a\e_1-\boldalpha)\e_1^\top, \quad \J = \begin{bmatrix}
      \boldzero & \I\\0 & \boldzero^\top
   \end{bmatrix}.
\end{align}
To obtain a discrete-time approximation of the system~\eqref{general-case-system}, we apply a first-order Lie–Trotter splitting method, see, e.g.,~\cite[Section 2.5]{Hairer2006}, where the $\boldGamma$-subsystem is discretized applying zero-order hold discretization, and the $\J\boldzeta$-subsystem is discretized via backward Euler scheme.
The $\boldGamma$ subsystem is linear in $\boldzeta$. 

Given the discrete state value $\boldzeta_k$ at index $k$, solving at first the $\boldGamma$-subsystem via exact discretization with sampling time $T$ yields a partial solution for the next sampling interval denoted by $\boldzeta_{k+1}^\prime$. This partial solution is computed as
\begin{align}\label{exact-Gamma-solution}
    &\boldzeta^\prime_{k+1} = \boldPhi(T) \boldzeta_k + \int_{kT}^{(k+1)T} \boldPhi\left((k+1)T-\tau\right) \cdot\nonum
    &\,\cdot\left((\boldbeta + a\boldtheta) r(\tau) - \boldtheta \dot r(\tau) + \boldtheta b u_k + \e_{p+1} \delta(\tau)\right) \,\d \tau,
\end{align}
with the %
matrix exponential function 
\begin{align}\label{Gamma-transition-matrix}
    \boldPhi(t) &= e^{\M t}.
\end{align}
Instead of solving the convolution integrals, we approximate the solution~\eqref{exact-Gamma-solution} via a zero-order hold discretization of the system.
Thus, over each sampling interval, we approximate the time-dependent input signals $r(t)$, $\delta(t)$, and $\dot r(t)$ via piecewise constant signals with the values
\begin{align}\label{rk}
&r_k = \begin{cases}
   \int_{kT}^{k^+T}\frac{r(t)}{T}\,\d t \quad\text{if this can be computed explicitly, }\\
   \frac{r(k^+T) + r(kT)}{2} \quad \text{otherwise, }
\end{cases} \\
\label{deltak}
&\delta_{k+1} = \frac{1}{T}\int_{kT}^{k^+T}\delta (t)\d t,
\end{align}
and $\dot r_k$ as in~\eqref{dotrk},
respectively. 
Note that $|\delta_{k+1}| \leq \bar D$ has the same bound as $\delta(t)$, and $\delta_{k+1}$ is unknown.
This approximation is valid, if the corresponding time signals are \emph{slow} enough, or, equivalently, the sampling time $T$ is \emph{small} enough.
This yields the approximative solution of the $\boldGamma$-subsystem being
\begin{align}\label{approximative-Gamma-solution}
    \hat\boldzeta^\prime_{k+1} &= \boldPhi(T) \boldzeta_k + \int_{0}^{T} \boldPhi(\tau) \,\d\tau \,\cdot\nonum
    &\cdot\left((\boldbeta + a\boldtheta) r_k - \boldtheta \dot r_k + \boldtheta b u_k + \e_{p+1} \delta_{k+1}\right).
\end{align}
Note that the system matrix $\M$ of the subsystem contains only zeros except in the first column. 
Therefore, exploiting the series definition of the matrix exponential, the matrix $\boldPhi(t)$ in~\eqref{Gamma-transition-matrix} at $t=T$ takes the form
\begin{align}\label{structure-Phi}
    \boldPhi(T) = \I - T\hat\boldalpha \e_1^\top,
\end{align}
with some vector $\hat\boldalpha$ that is known.
Due to this structure of $\boldPhi(t)$, and $\beta_1=0$ from~\eqref{beta1-0}, the vectors multiplied with the disturbance $\delta_{k+1}$ and the reference $r(t)$ in~\eqref{approximative-Gamma-solution} compute as
\begin{align}\label{disturbance-vector}
    \int_{0}^{T} \boldPhi(\tau) \e_{p+1} \,\d\tau &= T \e_{p+1},\quad
    \int_{0}^{T} \boldPhi(\tau) \boldbeta \,\d\tau = \boldbeta.
\end{align}
Furthermore, looking at the solution~\eqref{approximative-Gamma-solution}, we define the vectors
\begin{align}\label{definition-bar-vectors}
    \bar\boldtheta &= \frac{1}{T} \left( \int_{0}^{T} \boldPhi(\tau) \,\d\tau \right) \boldtheta,\nonum
    \bar\boldalpha &= \hat\boldalpha + a\bar\boldtheta.
\end{align}
Note that in general $\bar\theta_1 \neq 1$, compared to $\theta_1=1$ in the continuous-time system~\eqref{zeta-system}.
Thus, we can write the solution~\eqref{approximative-Gamma-solution} as
\begin{align}\label{partial-approximative-solution}
    \hat\boldzeta^\prime_{k+1} &= \boldzeta_k - T\bar\boldalpha \zeta_{1,k} + T\boldbeta r_k + T\e_{p+1} \delta_{k+1} + \nonum
    &\quad + T\bar\boldtheta (b u_k + a(\zeta_{1,k} + r_k) - \dot r_k).
\end{align}
Note that the solution~\eqref{partial-approximative-solution} has the structure of a forward Euler discretization scheme applied to the $\boldGamma$-subsystem in~\eqref{zeta-system}, only with differently computed vectors~\eqref{definition-bar-vectors}.

Subsequently, following the Lie-Trotter approximation scheme, we apply the $\J\boldzeta$-subsystem to the partial approximative solution~\eqref{partial-approximative-solution}.
For this second subsystem, we apply a backward Euler discretization scheme, giving the overall approximated solution for the next sampling interval as
\begin{align}\label{solution-J-subsystem}
    \boldzeta_{k+1} = \hat\boldzeta^\prime_{k+1} + T\J \boldzeta_{k+1}.
\end{align}
Combining the partial solutions~\eqref{partial-approximative-solution} and~\eqref{solution-J-subsystem}, and considering $\beta_1=0$ from~\eqref{beta1-0}, yields the discrete-time approximation of the continuous-time system~\eqref{zeta-system} as
\begin{align}\label{approximation-zetak}
&\zeta_{1,k+1} = \zeta_{1,k} + T\zeta_{2,k+1} - T\bar\alpha_1\zeta_{1,k} + \nonum
&\qquad + T\bar\theta_1 (b u_k + a (\zeta_{1,k} + r_k) - \dot r_k),\nonum
&\zeta_{i,k+1} = \zeta_{i,k} + T\zeta_{i+1,k+1} - T\bar\alpha_i\zeta_{1,k} + T\beta_i r_k + \nonum
&\qquad + T\bar\theta_i (b u_k + a (\zeta_{1,k} + r_k) - \dot r_k),\, i\in\{2,\dots,p\},\nonum
&\zeta_{p+1,k+1} = \zeta_{p+1,k} - T\bar\alpha_{p+1}\zeta_{1,k} + T\beta_{p+1} r_k + \nonum
&\qquad + T\theta_{p+1} (b u_k + a (\zeta_{1,k} + r_k) - \dot r_k) + T\delta_{k+1}.%
\end{align}
\begin{rem}
    The state $\zeta_{1,k}$ is corrected in every time step using the available measurement $\sigma(kT)$, by setting $\zeta_{1,k} = \sigma(kT)$. %
    The applied Lie-Trotter splitting approximation method is a first-order method, so the error in every step is of order $O(T^2)$.
\end{rem}

\section{Proper Implicit Discretization of the RESMC}\label{sec:implicit-discretization}

In this section, we derive an implicit discretization denoted as \emph{proper implicit discretization} of the continuous-time RESMC~\eqref{dynamic-controller},~\eqref{controller-functions-c},~\eqref{controller-functions-g}.
The \textit{proper implicit discretization} consists of the following steps: 
\begin{enumerate}
    \item Derive a discrete-time plant description that is in implicit Euler form and provides a good approximation of the sampled plant states (performed in Section~\ref{plant-discretization}).
    \item Setup controller structure with dynamics in implicit Euler form and  insert into the discrete-time system dynamics;
    \item Define a state transformation such that the transformed state in the next step is implicitly known in the current step;
    \item Design the controller terms such that the closed-loop dynamics of the transformed coordinates resemble the implicit Euler discretization of the continuous-time error dynamics.
\end{enumerate}

\begin{rem}
The discrete-time controller does not rely on a direct discretization of the controller itself.
Instead, it is reconstructed in discrete time based on a discretization of the desired closed-loop error system.
\end{rem}

\begin{rem}
The third step described above is the fundamental difference to other implicit discretizations in literature, see, e.g.,~\cite{Brogliato2020, Huber2016}.
There, the authors neglect unknown information to obtain an approximation for the state in the next sampling instance.
Such an approximation via neglecting unknowns is not necessary in the presented approach.
\end{rem}

\begin{rem}
All equations that must be implemented to obtain the proposed proper implicit RESMC are embedded in green boxes.
\end{rem}

The controller must be constructed based on the sampled output $y_k$ and the continuous-time reference $r(t)$ only.
With $y_k = y(kT) = \sigma(kT) + r(kT)$ from~\eqref{output-equation}, we know $\zeta_{1}(kT) = \sigma(kT) = y_k - r(kT)$. 
So, the controller has access to the error variable $\zeta_{1}(t)$ at the sampling instances $t=kT$.

Note the structural similarity of the discrete-time plants in the nomincal case~\eqref{discrete-time-nominal-plant} and in the general case~\eqref{approximation-zetak}.
As described in Section~\ref{nominal-plant-discretization}, the nominal case corresponds to $a = 0$, $\boldalpha = \boldbeta = \boldzero$.
Thus, in the nominal case the matrix $\M$ in the general system description~\eqref{general-case-system} results in $\M = \boldzero$ and the matrix~\eqref{Gamma-transition-matrix} becomes $\I$. 
Due to~\eqref{structure-Phi} the vector $\hat\boldalpha = \boldzero$.
This ultimately yields $\bar\boldtheta = \e_1$ and $\bar\boldalpha = \boldzero$ in~\eqref{definition-bar-vectors}.
By using $\bar\boldtheta = \e_1$, $\bar\boldalpha = \boldbeta = \boldzero$, the discrete-time plant in the nominal case~\eqref{discrete-time-nominal-plant} is equivalent to the general discrete-time plant~\eqref{approximation-zetak}.
Therefore, it is sufficient to derive the discrete-time controller for the general discrete-time system~\eqref{approximation-zetak}.

Similar to the continuous-time control law in~\eqref{continuous-time-control-law}, using~\eqref{rk},~\eqref{dotrk} and $\zeta_{1,k} = \sigma(kT)$, we define the discrete-time control law
\begin{framedBox}
\vspace{-1.3em}
\begin{align}\label{discrete-time-control-law}
    u_k = \frac{1}{b}\left( \bar u_k - a (\sigma(kT) + r_k) + \dot r_k \right),
\end{align}
\end{framedBox}
that compensates for the known plant dynamics.

No information is given about the disturbance states $\zeta_{2,k}$ up until $\zeta_{p+1,k}$ or the unknown input $\delta_k$.
The terms related to $\zeta_{1,k}$, $r_k$ and $\dot r_k$ in~\eqref{approximation-zetak}, however, are known and can be compensated.
Thus, we define the discrete-time controller in implicit structure copying the disturbance model part in the discrete-time system~\eqref{approximation-zetak} yielding
\begin{align}\label{discrete-time-general-controller}
   \bar u_k &= \frac{1}{\bar\theta_1} (v_{1,k+1} + c_{1,k} + \psi_{1,k}),\nonum
   v_{i,k+1} &= v_{i,k} + T v_{i+1,k+1} + T c_{i+1,k} + T \psi_{i+1,k},\, \nonum
   &i\in\{1,\dots,p-1\},\nonum
   v_{p,k+1} &= v_{p,k} + Tc_{p+1,k}  + T \psi_{p+1,k} ,
\end{align}
with discrete-time stabilizing controller terms $\psi_{i,k} $ and compensating controller terms $c_{i,k}$ %
to be designed.
Using $\zeta_{1,k} = \sigma(kT)$ we select the compensation terms as
\begin{framedBox}
\vspace{-1.3em}
\begin{align}\label{compensation-functions-general}
   c_{1,k} &= \bar\alpha_1\sigma(kT),\nonum
   c_{i,k} &= \bar c_{i,k} - \bar\theta_i\bar u_k, \text{ with}\nonum
   \bar c_{i,k} &= \bar\alpha_i\sigma(kT) - \beta_i r_k,\,i\in\{2,\dots,p+1\}.
\end{align}
\end{framedBox}

Combining the discrete-time system~\eqref{approximation-zetak} and the discrete-time controller~\eqref{discrete-time-general-controller},~\eqref{compensation-functions-general} with the closed-loop error states $\varepsilon_{i,k} = v_{i-1,k} + \zeta_{i,k}, i\in\{2,\dots,p+1\}$ and $\varepsilon_{1,k} = \zeta_{1,k}$ yields the error dynamics
\begin{align}\label{discrete-time-error-dynamics}
   \varepsilon_{i,k+1} &= \varepsilon_{i,k} + T\varepsilon_{i+1,k+1} + T \psi_{i,k} ,\,i\in\{1,\dots,p\}\nonum
   \varepsilon_{p+1,k+1} &= \varepsilon_{p+1,k} + T\psi_{p+1,k}  + T\delta_{k+1},\nonum
   \sigma(kT) &= \varepsilon_{1,k}.
\end{align} 
Finally, the discrete-time controller terms $\psi_{i,k} $ must be designed.
As a motivation we consider the implicit discretization of the continuous-time closed-loop error dynamics~\eqref{continuous-error-dynamics}, which are given by
\begin{align}\label{implicitly-discretized-error-dynamics}
   \varepsilon^*_{i,k+1} &= \varepsilon^*_{i,k} + T \varepsilon^*_{i+1,k+1} - Tk_i \sgnpow{\varepsilon^*_{1,k+1}}{\frac{p+1-i}{p+1}},\nonum
   &\forall i\in\{1,\dots,p\},\nonum
   \varepsilon^*_{p+1,k+1} &\in \varepsilon^*_{p+1,k} - T k_{p+1} \sgn{\varepsilon^*_{1,k+1}} + T\delta^*_{k+1},
\end{align}
with an unknown discrete-time disturbance $\delta^*_{k+1}$ originating from the continuous-time disturbance $\delta(t)$.
By comparing~\eqref{discrete-time-error-dynamics} and~\eqref{implicitly-discretized-error-dynamics} with $\varepsilon^*_{i,k} = \varepsilon_{i,k}$ we see that the terms $\sgnpow{\varepsilon^*_{1,k+1}}{\frac{p+1-i}{p+1}}$ would make up the discrete-time controller functions. 
Thus, $\varepsilon^*_{1,k+1}$ must be determined. 
This, however, is not possible, as solving~\eqref{implicitly-discretized-error-dynamics} for $\varepsilon^*_{1,k+1}$ contains the unknown errors $\varepsilon^*_{2,k}$ until $\varepsilon^*_{p+1,k}$.
Instead, the controller functions are designed such that the dynamics~\eqref{implicitly-discretized-error-dynamics} are recovered in a transformed state space.

We define the transformed error states by moving all future terms in~\eqref{discrete-time-error-dynamics} to the same side, and obtain
\begin{align}\label{transformed-step-1}
\propertransformedstate_{i,k+1} &= \varepsilon_{i,k+1} - T \varepsilon_{i+1,k+1} =\nonum
&= \varepsilon_{i,k} + T \psi_{i,k} ,\,i\in\{1,\dots,p\}\nonum
\propertransformedstate_{p+1,k+1} &= \varepsilon_{p+1,k+1} - T\delta_{k+1} = \nonum
&=\varepsilon_{p+1,k} + T\psi_{p+1,k},
\end{align}
and thus
\begin{align}\label{definition-eta}
\propertransformedstate_{i,k} &= \varepsilon_{i,k} - T \varepsilon_{i+1,k},\, i\in\{1,\dots,p\},\nonum
\propertransformedstate_{p+1,k} &= \varepsilon_{p+1,k} - T \delta_{k}.
\end{align}
The prediction $\propertransformedstate_{1,k+1} = \varepsilon_{1,k} + T \psi_{1,k} $ is therefore fully known after selecting the controller term $\psi_{1,k}$.
Inserting $\varepsilon_{i,k} = \propertransformedstate_{i,k} + T\varepsilon_{i+1,k}\,\forall i\in\{1,\dots,p\}$ and $\propertransformedstate_{p+1,k} = \varepsilon_{p+1,k} + T\delta_{k}$ from~\eqref{definition-eta} into~\eqref{transformed-step-1} leads to
\begin{align}\label{transformed-step-2}
\propertransformedstate_{i,k+1} &= \propertransformedstate_{i,k} + T\psi_{i,k}  + T\varepsilon_{i+1,k},\,\forall i\in\{1,\dots,p\},\nonum
\propertransformedstate_{p+1,k+1} &= \propertransformedstate_{p+1,k} + T\psi_{p+1,k}  + T\delta_{k}.
\end{align}
Further, using $\varepsilon_{i+1,k} = \propertransformedstate_{i+1,k+1} - T\psi_{i+1,k},\,\forall i\in\{1,\dots,p\}$, from~\eqref{transformed-step-1} yields the transformed error dynamics
\begin{align}\label{transformed-step-3}
\propertransformedstate_{i,k+1} &= \propertransformedstate_{i,k} + T\propertransformedstate_{i+1,k+1} +\nonum
 &+ T\psi_{i,k}  - T^2\psi_{i+1,k},\,\forall i\in\{1,\dots,p\},\nonum
\propertransformedstate_{p+1,k+1} &= \propertransformedstate_{p+1,k} + T\psi_{p+1,k}  + T\delta_{k}.
\end{align}
By designing the controller functions as
\begin{align}\label{discrete-time-controller-functions}
   \psi_{i,k}  &= - k_i \sgnpow{\propertransformedstate_{1,k+1}}{\frac{p+1-i}{p+1}} + T\psi_{i+1,k},\,\nonum
   &i \in \{1,\dots,p\},\nonum
   \psi_{p+1,k}  &\in -k_{p+1}\sgn{\propertransformedstate_{1,k+1}},
\end{align}
with $\propertransformedstate_{1,k+1}$ still to be determined,
we yield the closed-loop transformed error dynamics
\begin{align}\label{discrete-time-closed-loop-error-dynamics}\
\propertransformedstate_{i,k+1} &= \propertransformedstate_{i,k} + T\propertransformedstate_{i+1,k+1} - Tk_i \sgnpow{\propertransformedstate_{1,k+1}}{\frac{p+1-i}{p+1}},\nonum
&\forall i\in\{1,\dots,p\},\nonum
\propertransformedstate_{p+1,k+1} &\in \propertransformedstate_{p+1,k} - Tk_{p+1} \sgn{\propertransformedstate_{1,k+1}} + T\delta_{k}.
\end{align}
These error dynamics~\eqref{discrete-time-closed-loop-error-dynamics} are equivalent to the motivation~\eqref{implicitly-discretized-error-dynamics}, i.e. the implicit discretization of the continuous-time error dynamics of the RESMC~\eqref{continuous-error-dynamics}. 
However, different to~\eqref{implicitly-discretized-error-dynamics}, the controller leading to the error dynamics~\eqref{discrete-time-closed-loop-error-dynamics} can be implemented, as $\propertransformedstate_{1,k+1}$ is known at time step $k$, as shown in the following.

In order to implement the controller~\eqref{discrete-time-general-controller},~\eqref{compensation-functions-general},~\eqref{discrete-time-controller-functions}, $\propertransformedstate_{1,k+1}$ must be determined explicitly. Extending $\propertransformedstate_{1,k+1} = \varepsilon_{1,k} + T \psi_{1,k} $ from~\eqref{transformed-step-1} using~\eqref{discrete-time-controller-functions} gives
\begin{align}
   \varepsilon_{1,k} &\in
\propertransformedstate_{1,k+1} + Tk_1\sgnpow{\propertransformedstate_{1,k+1}}{\frac{p}{p+1}} + \dots\nonum
& \quad + T^{p+1}k_{p+1} \sgn{\propertransformedstate_{1,k+1}}.
\end{align}
Solving this inclusion using $\varepsilon_{1,k} = \sigma(kT)$ yields
\begin{framedBox}
\vspace{-1.3em}
\begin{align}\label{solution-eta1kp1}
   \propertransformedstate_{1,k+1} = \begin{cases}
   0, \qquad\quad \text{ if } |\sigma(kT)| \leq T^{p+1}k_{p+1},\\
   \sign{\sigma(kT)}\lambda^{p+1} \qquad \text{ else,}
   \end{cases}
\end{align}
\end{framedBox}
with $\psi_{p+1,k}  = -\frac{\sigma(kT)}{T^{p+1}}$ from~\eqref{discrete-time-controller-functions} in the first case, and $\lambda>0$ is the solution of the polynomial equation
\begin{framedBox}
\vspace{-1.3em}
\begin{align}\label{polynomial-equation-lambda}
   \lambda^{p+1} + Tk_1\lambda^p + \dots + T^{p+1}k_{p+1} = |\sigma(kT)|
\end{align}
\end{framedBox}
in the second case. By combining both cases, the last controller term computes as
\begin{align}\label{psi_p+1}
   \psi_{p+1,k} = -k_{p+1}\sat{\frac{\sigma(kT)}{T^{p+1}k_{p+1}}}.
\end{align}

\begin{lem}\label{lemma-unique-polynomial-solution}
Equation~\eqref{polynomial-equation-lambda} has a unique positive real solution $\lambda\in\mathbb{R}^+$ if $|\varepsilon_{1,k}| > T^{p+1}k_{p+1}$.
\end{lem}
\begin{pf}
   Denote $p(\bar\lambda) = \bar\lambda^{p+1} + Tk_1\bar\lambda^p + \dots + T^{p+1}k_{p+1}$ as the left-hand side of~\eqref{polynomial-equation-lambda}.
   Note that $p(\bar\lambda)$ is continuous.
   For $\bar\lambda > 0$ every term in $p(\bar\lambda)$ is positive and monotonically increasing. Thus, also the sum $p(\bar\lambda)$ is positive and monotonically increasing. We have ${p(0) = T^{p+1}k_{p+1} < |\varepsilon_{1,k}|}$ and $p(|\varepsilon_{1,k}|^{\frac{1}{p+1}}) > |\varepsilon_{1,k}|$. 
   Due to monotonicity and continuity of $p(\bar\lambda)$ the solution $\lambda\in\mathbb{R}^+$ for which $p(\lambda) = |\varepsilon_{1,k}|$ exists and is unique.
   \qed
\end{pf}
For implementation in higher-order cases, equation~\eqref{polynomial-equation-lambda} must be solved numerically.
Note that then the proof of Lemma~\ref{lemma-unique-polynomial-solution} provides a useful interval for initializing the root finding algorithm.

To yield an explicit formulation of the proper implicit RESMC, at first the controller~\eqref{discrete-time-general-controller},~\eqref{compensation-functions-general} must be solved for the implicitly appearing control value $\bar u_k$.
Repeatedly inserting $v_{i,k+1}$ for $i=1,2,\dots,p+1$ into the control law for $\bar u_k$ in~\eqref{discrete-time-general-controller} and using~\eqref{discrete-time-controller-functions} leads to the equation %
\begin{framedBox}
\vspace{-1.3em}
    \begin{align}\label{solution-uk-general}
   &\bar u_k (\bar\theta_1 + T\bar\theta_2 + \dots + T^p\bar\theta_{p+1}) = \nonum
   &= c_{1,k} + \begin{bmatrix}
   1 & T & \dots & T^{p-1}
   \end{bmatrix}\begin{bmatrix}
      v_{1,k} + T\bar c_{2,k}\\\vdots\\ v_{p,k} + T\bar c_{p+1,k}
   \end{bmatrix} + \nonum
   & + \begin{bmatrix}
      1 & 2T & \dots & (p+1)T^{p}
   \end{bmatrix}\begin{bmatrix}
      -k_1 \sgnpow{\propertransformedstate_{1,k+1}}{\frac{p}{p+1}}\\
      \vdots\\
      -k_p \sgnpow{\propertransformedstate_{1,k+1}}{\frac{1}{p+1}}\\
      -k_{p+1}\sat{\frac{\sigma(kT)}{T^{p+1}k_{p+1}}}
   \end{bmatrix},
   \end{align}
\end{framedBox}
with $\propertransformedstate_{1,k+1}$ as given in~\eqref{solution-eta1kp1}.
Using $\bar u_k$ from~\eqref{solution-uk-general} in~\eqref{compensation-functions-general} yields explicit solutions $c_{i,k}$. Repeatedly inserting $v_{i,k+1}$ for $i=2,\dots,p+1$ on the right-hand side of~\eqref{discrete-time-general-controller} leads to the explicit form of the dynamic part of the proper implicit RESMC being
\begin{framedBox}
\vspace{-1.3em}
    \begin{align}\label{implicit-controller-general}
   &\begin{bmatrix}
      v_{1,k+1}\\v_{2,k+1}\\\vdots\\v_{p,k+1}
   \end{bmatrix} = \begin{bmatrix}
      1 & T & \dots & T^{p-1}\\
      0 & 1 & \dots & T^{p-2}\\
      \vdots & & \ddots & \vdots\\
      0 & 0 & \dots & 1
   \end{bmatrix}\begin{bmatrix}
      v_{1,k} + Tc_{2,k}\\v_{2,k} + Tc_{3,k}\\\vdots\\v_{p,k} + Tc_{p+1,k}
   \end{bmatrix} + \nonum
   &+\begin{bmatrix}
      T & 2T^2 & \dots & pT^{p}\\
      0 & T & \dots & (p-1)T^{p-1}\\
      \vdots &   & \ddots & \vdots\\
      0 & 0 & \dots & T
   \end{bmatrix}\begin{bmatrix}
      -k_2 \sgnpow{\propertransformedstate_{1,k+1}}{\frac{p-1}{p+1}}\\
      \vdots\\
      -k_p \sgnpow{\propertransformedstate_{1,k+1}}{\frac{1}{p+1}}\\
      -k_{p+1}\sat{\frac{\sigma(kT)}{T^{p+1}k_{p+1}}}
   \end{bmatrix}.
   \end{align}
\end{framedBox}

Note that in the nominal case of stabilizing the output that is disturbed by an integrator chain, as described in the beginning of this section, we have $\bar\boldtheta = \e_1$, $\bar\boldalpha = \boldzero$, $\boldbeta = 0$.
This simplifies the controller~\eqref{implicit-controller-general},~\eqref{solution-uk-general},~\eqref{compensation-functions-general} to 

\begin{framedBox}
\vspace{-1.3em}
\begin{align}\label{implicit-controller-nominal}
   &\begin{bmatrix}
      T\bar u_k\\v_{1,k+1}\\v_{2,k+1}\\\vdots\\v_{p,k+1}
   \end{bmatrix} = \begin{bmatrix}
      T & T^2 & \dots & T^{p}\\
      1 & T & \dots & T^{p-1}\\
      0 & 1 & \dots & T^{p-2}\\
      \vdots & & \ddots & \vdots\\
      0 & 0 & \dots & 1
   \end{bmatrix}\begin{bmatrix}
      v_{1,k}\\v_{2,k}\\\vdots\\v_{p,k}
   \end{bmatrix} +  \nonum
   &+\begin{bmatrix}
      T & 2T^2 & \dots & (p+1)T^{p+1}\\
      0 & T & \dots & pT^{p}\\
      \vdots &    & \ddots & \vdots\\
      0 & 0  & \dots & T
   \end{bmatrix}\begin{bmatrix}
      -k_1 \sgnpow{\propertransformedstate_{1,k+1}}{\frac{p}{p+1}}\\
      \vdots\\
      -k_p \sgnpow{\propertransformedstate_{1,k+1}}{\frac{1}{p+1}}\\
      -k_{p+1}\sat{\frac{\sigma(kT)}{T^{p+1}k_{p+1}}}
   \end{bmatrix},
\end{align}
\end{framedBox}
with $\propertransformedstate_{1,k+1}$ as given in~\eqref{solution-eta1kp1}.

\begin{thm}\label{theorem1}
Consider the nominal case of the control problem~\eqref{output-equation},~\eqref{disturbance-model} defined by $a = 0$, $\f = \boldzero$, ${\S = \begin{bmatrix}
    \boldzero & \I\\0 &\boldzero^\top
\end{bmatrix}}$, $\boldgamma = \e_1$, and $\boldd = \e_p$, 
with the disturbance fulfilling $ {|\Delta(t)| \leq D}$.
Furthermore, consider the sampled-data controller~\eqref{zoh-control},~\eqref{implicit-controller-nominal}, \eqref{solution-eta1kp1} applied to~\eqref{output-equation},~\eqref{disturbance-model}, with $\zeta_{1,k} = \varepsilon_{1,k} = \sigma(kT)$.
Then, there exist controller gains $k_i,\,i\in\{1,\dots,p+1\}$ such that for any initial states $\sigma(0), \boldphi(0)$ the samples $\sigma(kT)$ converge to a vicinity of the origin in finite time, specifically
\begin{align}\label{thm-eq-1}
\exists k_c: |\sigma(kT)| \leq T^{p+1} bD\,\forall k\geq k_c.
\end{align}
Finally, in the undisturbed case, i.e. $\Delta(t) = 0\,\forall t$, the samples $\sigma(kT)$ converge to zero in finite time.
\end{thm}
\begin{pf}
System~\eqref{output-equation},~\eqref{disturbance-model} in this setting is equivalent to the transformed system~\eqref{nominal-zeta-system}. 
System~\eqref{discrete-time-nominal-plant} is a discrete-time representation of the continuous-time system~\eqref{nominal-zeta-system} fulfilling all properties in Lemma~\ref{lemma-nominal-discretization}, with $\delta(t) = b \Delta(t)$ and thus their bounds fulfill $b D = \bar D$.
For properly selected gains $k_i,\,i\in\{1,\dots,p+1\}$ the origin of the transformed error system~\eqref{discrete-time-closed-loop-error-dynamics} is finite-time stable. The proof is based on a discrete-time Lyapunov function and was developed in~\cite[Proposition 5.10]{Seeber_2025}.
Assume $\propertransformedstate_{i,k} = 0\,\forall i\in\{1,\dots,p+1\}, k\geq k_c$ for some $k_c\in\mathbb{Z}^+$.
From transformation~\eqref{definition-eta} and property (P3) of Lemma~\ref{lemma-nominal-discretization} follows that 
\begin{align}
|\varepsilon_{i,k}| \leq T^{p+2-i}\bar D\,\forall k \geq k_c, i\in\{1,\dots,p+1\}
\end{align}
and specifically~\eqref{thm-eq-1}, as $\varepsilon_{1,k} = \zeta_{1,k} = \zeta_1(kT) = \sigma(kT)$ due to property (P1) in Lemma~\ref{lemma-nominal-discretization}.
The last claim in the theorem is a direct consequence of~\eqref{thm-eq-1} with $D=0$.
\qed
\end{pf}

\begin{rem}
    In case $a \neq 0$, a similar result as in Theorem~\ref{theorem1} can be obtained by applying, e.g., the exact compensation technique presented by~\cite{Andritsch2025_Springer}.
\end{rem}

\begin{prop}
   If $|\zeta_{1,k}| \leq T^{p+1}\bar D$ holds, 
   the proper implicit RESMC~\eqref{implicit-controller-general},~\eqref{solution-uk-general},~\eqref{solution-eta1kp1},~\eqref{compensation-functions-general},~\eqref{discrete-time-control-law} is a deadbeat controller for the discrete-time system~\eqref{approximation-zetak}, i.e. the discrete-time closed-loop error dynamics~\eqref{discrete-time-error-dynamics} are a linear system with input $\delta_{k+1}$ and all its eigenvalues are zero. 
\end{prop}
\begin{pf}
    Note that $\zeta_{1,k} = \varepsilon_{1,k}$.
    Equations~\eqref{solution-eta1kp1},~\eqref{psi_p+1} for $|\zeta_{1,k}| \leq T^{p+1}\bar D$ yield $\propertransformedstate_{1,k+1} = 0$ and $\psi_{p+1,k} = -\frac{\varepsilon_{1,k}}{T^{p+1}}$.
    Inserting the controller functions~\eqref{discrete-time-controller-functions} into~\eqref{discrete-time-error-dynamics} in this case yields the closed-loop error dynamics
    \begin{align}
      \begin{bmatrix}
         \varepsilon_{1,k+1}\\
         \varepsilon_{2,k+1}\\
         \vdots\\
         \varepsilon_{p+1,k+1}\\
      \end{bmatrix} &= \begin{bmatrix}
         -p & T & \dots & T^p\\
         - \frac{p}{T} & 1 & \dots & T^{p-1}\\
         \vdots & &\ddots&\vdots\\
         - \frac{1}{T^p} & 0 & \dots & 1
      \end{bmatrix}
      \begin{bmatrix}
         \varepsilon_{1,k}\\
         \varepsilon_{2,k}\\
         \vdots\\
         \varepsilon_{p+1,k}\\
      \end{bmatrix} + \nonum
      &\quad +\begin{bmatrix}
         T^{p+1}&T^p&\dots&T
      \end{bmatrix}^\top\delta_{k+1}.
    \end{align}
    The system matrix $\M$ of these dynamics can be expressed as
    \begin{align}
        \A &= \begin{bmatrix}
            1 & T & \dots & T^p\\
            0 & 1 & \dots & T^{p-1}\\
            \vdots & & \ddots & \vdots\\
            0 & 0 & \dots & 1
        \end{bmatrix},\,\B = \begin{bmatrix}
            - (p+1) & 0 & \dots & 0\\
            -\frac{p}{T} & 0 & \dots & 0\\
            \vdots & & \ddots & \vdots\\
            -\frac{1}{T^p} & 0 & \dots & 0
        \end{bmatrix},\nonum
        \M &= %
        \A + \B.
    \end{align}
    With the vector $\v = \begin{bmatrix}
        -p & -\frac{p-1}{T} & \dots & -\frac{1}{T^p}
    \end{bmatrix}^\top$ and the shift matrix $\N = \begin{bmatrix}
        \boldzero & \I\\0 & \boldzero^\top
    \end{bmatrix}$,
    we have $\A = \sum_{i=0}^{p}(T\N)^i$, and $\B = \v\e_1^\top$.
    As $\N$ is nilpotent with $\N^{p+1} = \boldzero$, also $\A$ is nilpotent with $\A^{p+1} = \boldzero$ and furthermore, as $\A$ is a Neumann series, we have $\A = (\I-T\N)^{-1}$.
    We consider the similarity transformation
    \begin{align}\label{similarity-S}
        \S = \A^{-1}\M\A = \A + (\A^{-1}\v)(\e_1^\top\A),
    \end{align}
    and compute
    \begin{align}
        \A^{-1}\v &= (\I - T\N)\v = \begin{bmatrix}-1 & -\frac{1}{T} & \dots & -\frac{1}{T^p}\end{bmatrix}^\top,\nonum
        \e_1^\top\A &= \begin{bmatrix}1 & T & \dots & T^p\end{bmatrix}.
    \end{align}
    Therefore, $\S$ in~\eqref{similarity-S} computes as
    a strict lower triangular matrix with zeros on the diagonal.
    Consequently, $\S^{p+1} = \boldzero$ and furthermore due to the similarity transformation $\M^n = \A \S^n \A^{-1} = \boldzero$.
    Thus, $\M$ is nilpotent and all its eigenvalues are zero.
    \qed
\end{pf}

\begin{rem}
In the nominal case with the presented sample-wise exact discrete-time representation~\eqref{discrete-time-nominal-plant} of system~\eqref{nominal-zeta-system}, the input disturbance $\delta_{k+1}$ is defined via an integral containing values of $\zeta_2(t)$ for $t\in[(k-p)T, (k+1)T]$. 
Thus, consecutive values of $\delta_{k+1}$ are not independent.
For example, $\delta_{k+1}$ can not change from $-\bar D$ to $\bar D$ in one step.
This dependence could possibly be used to design a controller with a higher accuracy than the proposed controller for $p>1$.
\end{rem}

\begin{rem}
In the nominal case with $p=1$, the resulting proper implicit RESMC~\eqref{implicit-controller-nominal} coincides with the proper implicit discretization of the super-twisting controller (STC) proposed by~\cite{Seeber2025} and the modified implicit discretization of the STC presented by~\cite{Andritsch2024}.
\end{rem}

\section{Control Accuracy in the Nominal Case}\label{sec:control-accuracy}

Inspired from~\cite[Proposition~2.1]{Seeber2025} we formulate the following proposition that establishes a lower bound on the best possible worst-case control accuracy in the nominal case of stabilizing the system output that is disturbed by an integrator chain with unknown input.
\begin{prop}\label{prop:best-possible-accuracy}
   Consider the dynamics~\eqref{nominal-zeta-system} with bounded unknown input $\delta(t)$ fulfilling $|\delta(t)| \leq \bar D$, and the measurements $\zeta_1(kT)$ are the only known information from the system. 
   Assume that the control accuracy 
   \begin{align}
      ||\zeta_1(kT)|| \leq M,\,\forall k\geq k_c,
   \end{align}
   is achieved for some $k_c>0$.
   The accuracy $M$ of any causal control sequence $u_k$ and thus also the best possible worst-case accuracy is lower bounded by 
   \begin{align}\label{best-possible-accuracy}
      M \geq M_0 = \frac{3^{p+1}-1}{2^{p+1}(p+1)!}T^{p+1}\bar D.
   \end{align}
\end{prop}
\begin{pf}
   Consider the two disturbance signals
   \begin{align}
      \zeta^{i}_2(t) &= (-1)^{K-k-1}\frac{\bar D}{p!}\sgnpow{t-kT-\frac{T}{2}}{p},\,\nonum
      &\qquad\forall t\in[kT,(k+1)T),\,\forall k<K,\nonum
      \zeta^{i}_2(t) &= \frac{\bar D}{p!}\sgnpow{t-(K-1)T-\frac{T}{2}}{p},\,\forall t \geq KT,\nonum
      \zeta^{ii}_2(t) &= -\zeta^{i}_2(t),\,\forall t>0,
   \end{align}
   with some $K\geq k_c$.
   An example of these signals is depicted in Figure~\ref{fig:accuracy-bound}.
   Both signals, $\zeta^{i}_2(t)$ and $\zeta^{ii}_2(t)$, are continuous and have a $p$-th derivative existing almost everywhere that is bounded by $\bar D$.
   Thus, they can both be generated by the disturbance model $\zeta_2^{(p)}(t) = \delta(t)$ from~\eqref{nominal-zeta-system} with some input $\delta(t)$ and initial values $\zeta_2(0), \dots, \zeta^{(p-1)}_2(0)$.
   These two disturbances fulfill
   \begin{align}\label{special-disturbances-zero}
      \int_{kT}^{(k+1)T} \zeta_2^{i}(\tau) \d \tau = 0,\quad\forall k<K,\nonum
      \int_{kT}^{(k+1)T} \zeta_2^{ii}(\tau) \d \tau = 0,\quad\forall k<K,
   \end{align}
   Furthermore, these disturbance signals fulfill
   \begin{align}\label{special-disturbances-not-zero}
      &\int_{KT}^{(K+1)T} \zeta_2^{i}(t) \d t = -\int_{KT}^{(K+1)T} \zeta_2^{ii}(t) \d t = \nonum
      &= \int_{0}^{T}  \frac{\bar D}{p!}\left(t+\frac{T}{2}\right)^p \d t = \frac{3^{p+1}-1}{2^{p+1}(p+1)!}T^{p+1}\bar D = M_0.
   \end{align}
   With any control sequence $u_k$ the same samples $\zeta_1(kT)$ are achieved with $\zeta_2(t) = \zeta_2^{i}(t)$ and $\zeta_2(t) = \zeta_2^{ii}(t)$ for ${k<K}$ due to~\eqref{special-disturbances-zero}.
   Consider the solution~\eqref{solution-zeta1kp1} at the sampling instance $t=(K+1)T$.
   If $\zeta_1(KT) + b u_k \geq 0$, then ${\zeta_1((K+1)T) \geq M_0}$ in case $\zeta_2(t) = \zeta_2^{i}(t)$ due to~\eqref{special-disturbances-not-zero}. 
   Otherwise, if $\zeta_1(KT) + b u_k < 0$, then $\zeta_1((K+1)T) \leq -M_0$ in case $\zeta_2(t) = \zeta_2^{ii}(t)$.
   Thus, $M_0$ is a lower bound for the best possible worst-case accuracy.
   \qed
\end{pf}

\begin{figure}
   \centering
   \includegraphics[width=1\linewidth]{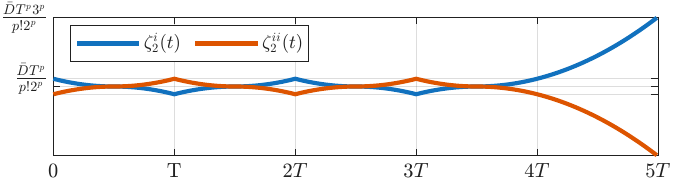}
   \caption{Disturbance signals $\zeta_2^{i}(t)$ and $\zeta_2^{ii}(t)$ from the proof of Proposition~\ref{prop:best-possible-accuracy} for $p=2$ and $K=4$.}\label{fig:accuracy-bound}
\end{figure}

For orders $p\in\{0,1,\dots,6\}$ the lower bounds $M_0$ according to Proposition~\ref{prop:best-possible-accuracy} rounded to two decimals are listed in Table~\ref{tab:accuracy-bounds}.
\begin{table}
   \centering
   \begin{tabular}{|c | c | c| c| c| c| c| c|}
      \hline
      $p$ & 0 & 1 & 2 & 3 & 4 & 5 & 6
      \\\hline
      $M_0/(T^{p+1}\bar D)$ & 1 & 1 & 0.54 & 0.21 & 0.06 & 0.02 & 0.00\\
      \hline
   \end{tabular}
   \caption{Lower bounds $M_0$ on the best possible worst-case accuracy for small orders $p$.}\label{tab:accuracy-bounds}
\end{table}

\begin{rem}
   The accuracy of the proposed discrete-time controller proved in Theorem~\ref{theorem1} is of order $T^{p+1}$ and scales with the disturbance bound $D$.
   This is a condition for a proper implicit discretization, as it was considered by~\cite{Seeber2023}. 
   Also, it is a better accuracy than achieved with explicit Euler discretization, which does not scale with~$D$, as shown by~\cite{Levant2011,Livne_2014}.
   According to Proposition~\ref{prop:best-possible-accuracy}, this accuracy is the best possible accuracy that can be achieved with a discrete-time controller if $p\leq 1$.
   For $p=1$ this was also shown in~\cite[Proposition 2.1]{Seeber2025}.
   However, for $p>1$, according to Proposition~\ref{prop:best-possible-accuracy}, there might exist a discrete-time controller that achieves a smaller accuracy band than $T^{p+1}D$.
\end{rem}

\section{Simulation Examples}\label{sec:simulations}

In this section we demonstrate the effectiveness of the proposed proper implicit discretization scheme of the RESMC in multiple examples.
The in Section~\ref{sec:implicit-discretization} proposed proper implicit RESMC (denoted as \emph{implicit RESMC}) is compared with %
an explicit Euler discretized RESMC (denoted as \emph{explicit RESMC}).
The explicit RESMC is constructed with a similar compensation of the non-nominal 
dynamics as in the proper implicit RESMC, that is outlined in the following.
We used the same approximative solution~\eqref{approximative-Gamma-solution} of the $\boldGamma$-subsystem in~\eqref{general-case-system}.
However, different to the backward Euler approach used in~\eqref{solution-J-subsystem}, we apply a forward Euler scheme for the $\J\boldzeta$-subsystem in~\eqref{general-case-system}, which gives
$%
    \boldzeta_{k+1} = \hat\boldzeta^\prime_{k+1} + T\J \boldzeta^\prime_{k+1}.
$
The controller was then redesigned in discrete-time to reproduce the explicit Euler discretized continuous-time error dynamics~\eqref{continuous-error-dynamics}.%

The system to be controlled consisting of the output dynamics~\eqref{output-dynamics} and the disturbance model~\eqref{disturbance-model} was simulated with a variable-step size solver using a maximum step size of $T/100$, and the disturbance model was initialized at $\varphi(0) = 0$ in all following simulations.

As a proof of concept of the proposed RESMC, and to demonstrate the analytical results from Section~\ref{sec:implicit-discretization}, the following high-order simulation was performed. 
We consider the nominal case, which consists of the output dynamics~\eqref{output-dynamics} with $a=0$, $b=1$, that are disturbed by an integrator chain with length $p=6$.
The discretization time was set to $T=0.1$, and
the bound on the disturbance was set to $D = 10$.
The given case of the integrator chain with $b=1$ yields $\mu = 1$, i.e., $\bar D = D$.
The plant was initialized at $x_1(0) = 1$.
The reference was selected zero in this example.
Finally, the disturbance $\delta(t)$ was once set to zero, and once selected as the rectangular signal $\delta(t) = \bar D \sign{\cos(2t)}$.
Note that this rectangular signal leads to a diverging disturbance signal $w(t)$, and thus also the control signal $u(t)$ must diverge to counteract the disturbance.
The controller gains were set to
$%
\k^\top = \begin{bmatrix}
    13.90 & 92.08 & 295.31 & 380.07 & 226.08 & 71.32 & 11.00
\end{bmatrix}$, following the suggested tuning proposed by~\cite{Levant2018}.
\begin{figure}[h!]
    \centering
    \includegraphics[width=1\linewidth]{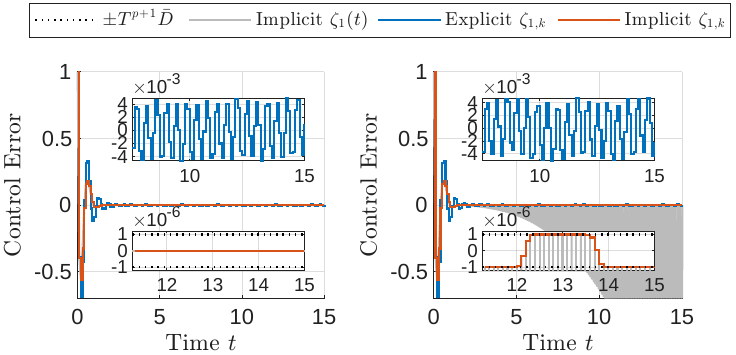}
    \caption{Controlled nominal system with $p=6$, no disturbance (left) and rectangular disturbance (right).}
    \label{simulation-nominal-system}
\end{figure}
The results are shown in Fig.~\ref{simulation-nominal-system}.
Pay attention to the different scaling factors in the zoomed portions of the results.
It is clearly visible that the \emph{explicit RESMC} leads to significant discretization chattering for both disturbances.
In the case with a disturbance, the amplitude of the remaining control error with \emph{explicit RESMC} is larger by a factor of $10^3$ compared to the \emph{implicit RESMC}.
With the diverging disturbance a significant control deviation occurs in between sampling instances with both controllers.
This is inevitable due to the piecewise constant control strategy and the diverging disturbance signal $w(t)$. 
However, when evaluating the control error $\zeta_1(t)$ at the sampling instances $kT$, this deviation is successfully mitigated.
The sampled errors $\zeta_{1,k}$ with the \emph{implicit RESMC} in this nominal case stay within the bounds proved in Theorem~\ref{theorem1}, also with the rectangular disturbance.

A second system considered in simulation is the mechanical system of rotating inertias that was described by~\cite{AndritschCDC2025}, to which the reader is referred for details.
The system parameters in~\eqref{output-dynamics},~\eqref{disturbance-model} take the form
\begin{align}
    a &= -1/3,\,\, b = 1/3,\,\, p=2, \,\, \boldgamma = \begin{bmatrix}
        1&0
    \end{bmatrix}^\top,\nonum
    \S &= \begin{bmatrix}
        -3/2 & 1\\2 & -2
    \end{bmatrix},\,\, \f = \begin{bmatrix}
        1/2\\0
    \end{bmatrix},\,\, \boldd = \begin{bmatrix}
        0\\1
    \end{bmatrix}.
\end{align}
The discretization time was set to $T=0.1$, and
the bound on the disturbance was set to $D = 2$.
The given system parameters yield $\mu = \frac{1}{3}$, i.e. $\bar D = \frac{1}{3}D$.
The plant was initialized at $x_1(0) = 1$.
The reference was selected as $r(t) = \sin(t)$ in this example.
Finally, the disturbance $\delta(t)$ was once set to zero, and once selected as the rectangular signal $\delta(t) = D \sign{\cos(2t)}$.
The controller gains were selected as
$\k^\top = \begin{bmatrix}
    2.20 & 2.57 & 1.47
\end{bmatrix}$, following the suggested tuning proposed by~\cite{Levant2018}.

\begin{figure}[h!]
    \centering
    \includegraphics[width=1\linewidth]{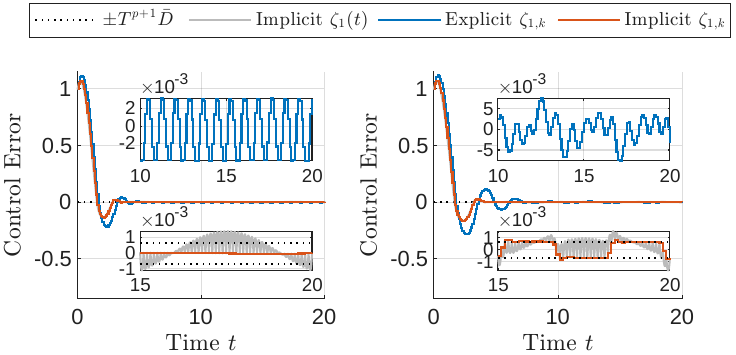}
    \caption{Reference tracking error of the system of rotating inertias, no disturbance input (left), and rectangular disturbance input (right).}
    \label{simulation-general-1}
\end{figure}

\begin{figure}[h!]
    \centering
    \includegraphics[width=1\linewidth]{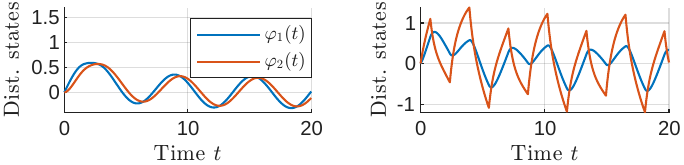}
    \caption{Disturbance states acting in the system of rotating inertias with \emph{implicit RESMC}, no disturbance input (left) and rectangular disturbance input (right).}
    \label{simulation-general-2}
\end{figure}

The results are depicted in Figs.~\ref{simulation-general-1} and~\ref{simulation-general-2}.
Again, also in this general case, the system under \emph{explicit RESMC} exhibits significant discretization chattering.
The remaining control error at the sampling instances with the \emph{implicit RESMC} almost stays within the proved bounds of the nominal system from Theorem~\ref{theorem1}, and does not exhibit high-frequency chattering.
As the disturbance model is coupled with the plant and the reference ${r(t) \not \equiv 0}$, there is a non-zero disturbance $w(t) = \boldgamma \boldphi(t)$ active in the system even in the case of the disturbance input $\delta(t) \equiv 0$.

\section{Conclusion}\label{sec:conclusion}

In this paper we presented a robust exact sliding-mode controller together with its proper implicit discretization that is able to reject disturbances generated by an LTI model with bounded input. %
The discretization is obtained by redesigning the controller in discrete time so that the implicit‑Euler closed‑loop error dynamics of the continuous‑time design are reproduced in transformed coordinates.
In the nominal case the accuracy of the discretization is related to the best possible worst-case accuracy by a factor depending only on the disturbance model order, independent of the sampling time.
Simulations confirm the theory.
Contrary to common belief, very small discretization times are not strictly necessary when sliding-mode control exploits advanced control concepts %
and implementation schemes.

\bibliographystyle{elsarticle-harv} 
\bibliography{bib}

\end{document}